\documentclass[11pt]{article}
\usepackage[utf8]{inputenc}

\usepackage[bookmarks,colorlinks,breaklinks]{hyperref}
\hypersetup{urlcolor=blue, colorlinks=true, citecolor=green!50!black, linkcolor=blue}
\usepackage[letterpaper, left=1in, right=1in, top=0.9in, bottom=0.9in]{geometry}
\usepackage[utf8]{inputenc}
\usepackage[american]{babel}

\usepackage[normalem]{ulem}

\usepackage{graphicx}
\graphicspath{{images/}{../images/}}

\usepackage{amsmath, amssymb, cases}
\usepackage{theoremstyles}
\usepackage{cleveref}
\usepackage{enumitem}
\usepackage{mdframed}
\usepackage{bbm}
\usepackage{bm}

\usepackage[ruled]{algorithm}
\usepackage[]{algpseudocode}

\usepackage{microtype}

\usepackage{mdframed}
\usepackage{xcolor}

\usepackage{modletters}

\usepackage{mathtools}

\newcommand\restr[2]{{% we make the whole thing an ordinary symbol
  \left.\kern-\nulldelimiterspace % automatically resize the bar with \right
  #1 % the function
  \vphantom{\big|} % pretend it's a little taller at normal size
  \right|_{#2} % this is the delimiter
  }}

\makeatletter
\def\moverlay{\mathpalette\mov@rlay}
\def\mov@rlay#1#2{\leavevmode\vtop{%
   \baselineskip\z@skip \lineskiplimit-\maxdimen
   \ialign{\hfil$\m@th#1##$\hfil\cr#2\crcr}}}
\newcommand{\charfusion}[3][\mathord]{
    #1{\ifx#1\mathop\vphantom{#2}\fi
        \mathpalette\mov@rlay{#2\cr#3}
      }
    \ifx#1\mathop\expandafter\displaylimits\fi}
\makeatother

\usepackage{graphicx}

\newcommand{\size}[1]{\mathrm{size}}

\newcommand{\set}[2][ ]{\{#2 \ifthenelse{\equal{#1}{ }}{ }{~|~#1}\}}

\newcommand{\comment}[1]{}

\newcommand{\seepage}[2][See]{
    \marginnote{
        \scriptsize {#1} p.~\pageref{#2}
    }
}

\newcommand{\reuse}[1]{
	\expandafter\stepcounter{#1_help}
    \expandafter\label{#1_app}
    \csname#1\endcsname*
}

\def\polylog{\operatorname{polylog}}

\def\Set#1{\left\{ #1 \right\}}
\def\Abs#1{\left| #1 \right|}

\def\Paren#1{\left( #1 \right)}		% need better macro name!
\usepackage{xcolor}

\newcommand{\isocut}{isolating sets}
\ifdefined\ShowComment
\newcommand{\blikstad}[1]{{\bf \color{magenta} [JOAKIM: #1]}}
\newcommand{\sagnik}[1]{{\bf \color{green} SAGNIK: #1}}
\newcommand{\danupon}[1]{{\bf \color{blue} DN: #1}}
\else
\newcommand{\danupon}[1]{}
\newcommand{\sagnik}[1]{}
\newcommand{\blikstad}[1]{}
\fi

\title{A Note on Isolating Cut Lemma for Submodular Function Minimization}
\author{Sagnik Mukhopadhyay \thanks{KTH Royal Institute of technology, Sweden, \texttt{sagnik@kth.se}} \and Danupon Nanongkai \thanks{KTH Royal Institute of technology, Sweden, \texttt{danupon@kth.se}}}
\date{}

\begin{document}

% \begin{titlepage}
% 	\maketitle
% 	\pagenumbering{roman}
% 	%\vspace{-.7cm}
% 	\input{abstract}
% 	%\newpage
% 	\setcounter{tocdepth}{2}
% 	%\tableofcontents
% 	%\newpage
% 	%\listoftheorems
% \end{titlepage}

\maketitle 
\begin{abstract}
It has been observed independently by many researchers that  the isolating cut lemma of Li and Panigrahi [FOCS 2020] can be easily extended to obtain new algorithms for finding the non-trivial minimizer of a symmetric submodular function and solving the hypergraph minimum cut problem. This note contains these observations. 
\end{abstract}

%This note collects some simple observations about the applications of the isolating cut lemma of Li and Panigrahi [FOCS 2020]. These applications include the problem of finding the non-trivial minimizer of a symmetric submodular function and hypergraph minimum-cut. These observations have been observed independently by many groups. 

%We show a generalization of the isolating cut lemma of Li-Panigrahi [FOCS 2020] to the problem of finding the non-trivial minimizer of a symmetric submodular function. As applications, we show how this gives bounds on the query complexity of the said problem and also time complexity of hypergraph minimum-cut.

%\newpage
\pagenumbering{arabic}
\section{Introduction}
Finding the minimum cut of a graph is a fundamental graph problem that has been at the center of research focus for many decades. In a recent result \cite{LiP2020}, Li-Panigrahi showed a deterministic algorithm of complexity\footnote{We use $m$ and $n$ to denote the number of edges and the number of vertices respectively.} $O(m \cdot \min(\sqrt m, n^{2/3}))$ for solving the minimum cut problem on weighted graphs. Their algorithm uses the max-flow algorithm (or, equivalently, the $st$-minimum cut algorithm) as a black box, and they show how to design an algorithm for global minimum cut for weighted graphs using $\polylog n$ many calls to a max-flow subroutine. The key insight in their work is the \textit{isolating cut lemma}.

Since their work, it has been observed by many researchers that the  \textit{isolating cut lemma} of \cite{LiP2020} can be extended to any arbitrary symmetric submodular function $f$. As such, this lets us obtain algorithms for finding a non-trivial minimizer for symmetric submodular functions with $\polylog n$ calls to an algorithm for submodular function minimization. This note contains formal arguments and statements that can
%In this note, we make this argument formal and state it in a way that can
be used to plug-in the best possible results for submodular function minimization in the future. We also sketch the application of such generalization to the query complexity of submodular function minimization and time complexity of finding a minimum cut in hypergraphs. We note that a recent work of Chekuri and Quanrud \cite{chekuri2021} has significant overlap with the content of this paper. 
%also generalize the isolating cut lemma to symmetric bisubmodular functions (which also capture symmetric submodular functions).

\paragraph{Organization.} In Section \ref{sec:extension}, we extend the notion of \textit{isolating cut} to symmetric submodular functions and state a general theorem as to how to compute non-trivial minimizer of such a function using blackbox access to an algorithm for submodular function minimization. In Section \ref{sec:sfm-query}, we apply the result from Section \ref{sec:extension} to compute the query complexity of finding a non-trivial minimizer. Finally, in Section \ref{sec:hg-mincut}, we show its application to finding minimum cut in a hypergraph.

\section{Extension to submodular function minimization} \label{sec:extension}
Consider a symmetric submodular function $f: 2^V \to \bbR$. We are interested in finding a non-trivial minimizer of $f$ which has size at most $n/2$. Note that, because $f$ is symmetric, such a minimizer will always exists. From now onwards, we drop the specification on  the size bound of the non-trivial minimizer---we assume that these are implied throughout. In this section, we prove the following theorem.

\begin{theorem}\label{thm:sfm-nontrivial}
Let $f$ be a symmetric submodular function and $k$ be a $1.5$-approximation\footnote{The constant $1.5$ is arbitrary. Any constant strictly less than $2$ is enough for our purpose.} of the size of its non-trivial minimizer. Also, assume that there is an algorithm $\cA$ for submodular function minimization.
%with query complexity $O(|V \setminus \{S \cup T\}|^c)$ for some $c \geq 1$, 
Then there is a randomized algorithm $\cB$ for computing the non-trivial minimizer of $f$ that \begin{enumerate}
    \item makes $O(\log n\log(n/k))$ many calls to $\cA$ with ground set size $O(|V|)$, and
    \blikstad{can the $\log(n/k)$ factor be removed by noting that there is a non-trivial minimizer of size $\ge n/2$, so we may assume that $k = \Omega(n)$? actually, picking $k=n/2$ always works, right? so no need of assuming uou have a constant approx}
    \item  makes an additional $O(\log n)$ many groups of $O(n/k)$ many calls to $\cA$ such that each group of calls is made with ground set $U_1, \cdots, U_{|R|}$ respectively and $\sum_i \Abs{U_j} \leq \Abs V$.
\end{enumerate}
$\cB$ outputs a subset of $V$ which is the non-trivial minimizer with high probability\footnote{By high probability, we mean with probability $1 - n^{-c}$ for some constant $c \geq 1$.}.
\end{theorem}

\begin{remark}
Without a constant approximation of the size of the non-trivial minimizer, the algorithm $\cA$ of \Cref{thm:sfm-nontrivial} still outputs a set. But there is no guarantee that it is a non-trivial minimizer.
\blikstad{One can try $k$ as powers of $2$ if one does not have a constant approximation, at the cost of an extra log-factor?}
\end{remark}

The rest of the section proves Theorem \ref{thm:sfm-nontrivial}. Extending the notion of minimum isolating cuts for the minimum cut problem \cite{LiP2020}, we define the minimum isolating set for $f$ in the following way.

\begin{definition}[Minimum \isocut] \label{def:isocut}
Consider a subset $R \subseteq V$ of the ground set $V$. The \textit{minimum \isocut} w.r.t. $R$ is a family of set $\{S_v: v \in R\}$ such that:\begin{enumerate}
    \item For every element $v \in R$, $S_v \cap R = \{v\}$, and \label{itm:uniq-intsec}
    \item Among all sets $S'_v$ with $S'_v \cap R = \{v\}$, $f(S_v)$ is the minimum and the set $S_v$ is inclusion-wise minimal among those sets. \label{itm:minimizer}
\end{enumerate}
\end{definition}

\noindent
The main technical theorem of this note is the following which is a generalization of Theorem II.2 of \cite{LiP2020}.
\begin{theorem}[Isolating set lemma] \label{thm:sfm-isolating}
Let $f$ be a symmetric submodular function and consider a subset $R \subseteq V$ of the ground set $V$. Also, assume that there is an algorithm $\cA$ for submodular function minimization.
%with query complexity $O(|V \setminus \{S \cup T\}|^c)$ for some $c \geq 1$, 
Then there is an algorithm $\cB$ for computing minimum isolating sets w.r.t. $R$ that \begin{enumerate}
    \item makes $\log |R|$ many calls to $\cA$ with ground set size $O(|V|)$, and
    \item makes an additional $|R|$ many calls to $\cA$ with ground set $U_1, \cdots, U_{|R|}$ respectively such that $U_j's$ are pairwise disjoint and $\sum_i \Abs{U_j} \leq \Abs V$.
\end{enumerate}
\end{theorem}

Before proving \Cref{thm:sfm-isolating}, we first outline how to prove \Cref{thm:sfm-nontrivial} using \Cref{thm:sfm-isolating}.

\begin{proof}[Proof of \Cref{thm:sfm-nontrivial}]
We sample each element of the ground set $V$ with probability $1/k$. Let the sampled set of elements be $R$. We show the following claim.

\begin{claim} \label{clm:uniq-inclusion}
With constant probability, $R$ has exactly one element from the non-trivial minimizer of $f$.
\end{claim}

To see why the claim is true, let us calculate the probability of choosing exactly one element from a set of size $k' = 2k/3$. The probability is at least

\begin{align*}
    k'\cdot \frac{1}{k}\Paren{1 - \frac 1 k}^{k'} \approx \frac 2 3 \cdot e^{- k'/k} = \frac 2 3\cdot e^{-\frac 2 3} > 0.
\end{align*}

We also need an additional claim which is as follows.

\begin{claim}\label{clm:inclusion}
With constant probability, $R$ includes at least one element from the complement set of the non-trivial minimizer of $f$.
\end{claim}

This is where we use the fact that the size of the non-trivial minimizer $k' \leq n/2$ which implies that $k \leq 3n/4$. This means that the complement set of the non-trivial minimizer has at least $n/4$ elements. The probability that $R$ includes at least one of them is at least
\[
1 - \Paren{1 - \frac 1 k}^{n/4} \approx 1 - e^{-\frac{n}{4k}}> 0.
\]

Note that the events of Claim \ref{clm:uniq-inclusion} and \ref{clm:inclusion} are independent of each other, and hence both of the occur simultaneously with constant probability. Repeating this sampling $O(\log n)$ times, we can get a sample $R$ such that it has (i) exactly one element from the non-trivial minimizer, and (ii) at least one element from the complement set of the non-trivial minimizer with very high probability. 

Let that unique element in $R$ be $v$. Once such a set $R$ is obtained, $\cB$ invokes \Cref{thm:sfm-isolating}. It is immediate that $S_v$ will be the non-trivial minimizer of $f$. So $\cB$ reports the $S_v$ that has the minimum value of $f(S_v)$ among the isolating sets.
\blikstad{I mostly agree with this argument, but there are some details missing. We need $|R| > 1$ for us to guarantee that we get a non-trivial minimizer (else we might obtain $S_v = V$). But in the case when $k = k' = n-1$ for example, the probability that both $|R\cap(\text{non-triv minimizer})| = 1$ and $|R| > 1$ is too small (something like $1/n$ ? ). \sout{One way to fix this is that we can always assume $k' \le n/2$ (as the complement to the non-trivial minimizer is also a minimizer by symmetry)}. Actually, even for $k' = n/2$ we might have the problematic case $k = n$, since $k$ is still indeed a constant approximation of $k'$ there.} \sagnik{Very nice catch. Fixed.}
\end{proof}

Towards proving \Cref{thm:sfm-isolating}, we define the notion of $s,t$-submodular function in the following way:

\begin{definition}[$st$-submodular function]
Given a submodular function $f$ and two disjoint subsets $S, T$ of $V$, the $st$-submodular function $f_{ST}$ w.r.t. to $f: 2^{V \setminus \{S \cup T\}} \to \bbR$ as follows: For any $A \subseteq V \setminus \{S \cup T\}$,
\begin{align*}
    f_{S, T}(A) = f(S \cup A).
\end{align*}
\end{definition}

 It is not hard to see that for any $S$ and $T$, $f_{S,T}$ is a submodular functions. Note that an oracle for any $f_{S,T}$ can be simulated by an oracle of $f$. Also, by setting $S = T = \emptyset$, we get back $f_{S,T} = f$. 
%The crucial observation is the following:

% \begin{observation} \label{obs:st-nontrivial}
% Let $f$ be a (possibly trivial) submodular function. For any $S \neq \emptyset$ and for any $T$, $f_{S,T}$ is a non-trivial submodular function.
% \end{observation}

% Observation \ref{obs:st-nontrivial} lets us assume that the algorithm $\cA$ can find the minimizer for any such $f_{S,T}$. 
Now we are ready to prove the main theorem, the proof of which is very similar to that of Theorem II.2 in \cite{LiP2020}. The algorithm $\cB$ performs the following steps:
\begin{description}
\item[Step 1.\label{itm:step-1}] The algorithm $\cB$ considers $\log |R|$ many bipartition of the set $R$ where, in the $i$-th bipartition $(S_i, T_i)$ of $R$, all elements of $R$ with $i$-th entry in their binary representation being 0 are in $S_i$. The algorithm $\cB$ then finds minimizers for $f_{S_i,T_i}$ for each $i$ by calling $\cA$ $\log \Abs R$ many times.

Every such call to $\cA$ produces a bipartition of the ground set $V$ (of which one part is the corresponding minimizer itself). We denote by $C_i(v)$ to be the part in the bipartition that contains $v$. Let $U_v = \bigcap_i C_i(v)$.

\item[Step 2.]\label{itm:step-2} For each $U_v$, denote $T_v = V \setminus U_v$, and call $\cA$ to find the minimizer of $f_{\{v\},T_v}$. The minimizers found in this step are $\{S_v: v \in R\}$.
\end{description}

It is easy to see that Step 1 makes $\log |R|$ many calls to $\cA$ and Step 2 makes $\Abs R$ calls to $\cA$. For Step 2, note further that each element of the ground set is included in exactly one $U_v$. Hence $\cB$ satisfies both conditions in Theorem \ref{thm:sfm-isolating}. Towards showing the correctness of $\cB$, we make the following claim.

\begin{claim} \label{clm:correctness}
Consider the set $S_v$ from Definition \ref{def:isocut} which satisfies both Condition \ref{itm:uniq-intsec} and \ref{itm:minimizer}. Then $S_v \subseteq U_v$.
\end{claim}

Note that, given Claim \ref{clm:correctness}, the correctness of the algorithm follows immediately as Step 2 correctly finds $S_v$ for all $v \in R$.  Next we prove Claim \ref{clm:correctness}.

\begin{proof}[Proof of Claim \ref{clm:correctness}]
Consider the $i$-th iteration in Step 1 of the algorithm. We show that $S_v \subseteq C_i(v)$ for all $i$. This implies that $S_v \subseteq U_v$ as $U_v = \bigcap_i C_i(v)$.

We start with the assumption that $S_v \not\subseteq C_i(v)$, and we show that this yields a contradiction. From the definition of submodularity, we know that
\[
f(S_v) + f(C_i(v)) \geq f(S_v \cup C_i(v)) + f(S_v \cap C_i(v)).
\]
Because $S_v \not\subseteq C_i(v)$, we have $S_v \cap C_i(v)$ is a strict subset of $S_v$. Also, note that $S_v \cap C_i(v)$ satisfies Condition \ref{itm:uniq-intsec}, i.e., it contains only the element $v$ from the set $R$. This, combined with the fact that $S_v$ is the minimizer of $f$ among all sets that satisfies Condition \ref{itm:uniq-intsec} and is also inclusion-wise minimial (Condition \ref{itm:minimizer}), we have $f(S_v) < f(S_v \cap C_i(v))$. So, for the above inequality to hold, it must be the case that $f(C_i(v)) > f(S_v \cup C_i(v))$.
\blikstad{We don't know that $f(S_v)$ is strictly less than $f(S_v\cap C_i(v))$, they could be equal? If I recall correctly, the original paper also assumes $S_v$ to be of minimum size (as a tie-breaker), and then we do indeed get the strict inequality.} \sagnik{Good point. Done.}

Because $f$ is a symmetric submodular function and $C_i(v)$ is obtained by calling $\cA$ on $f_{S_i,T_i}$, $C_i(v)$ is the minimizer of $f$ among all sets that contain $C_i(v)$ and does not contain $T_v$. \blikstad{\ldots and do not contain $T_i$?} That is $f(C_i(v)) \leq f(S_v \cup C_i(v))$ which is a contradiction.
\end{proof}

\section{Application I: SFM query complexity} \label{sec:sfm-query}

\begin{lemma}\label{lem:sfm-nontrivial-query}
The randomized query complexity of finding a non-trivial minimizer of a symmetric submodular function is $O(n^2 \log(nM)\log^3 n)$ where $M$ is the maximum absolute value of the function values.
\end{lemma}

\begin{remark}
For sub-exponential $M$, Lemma \ref{lem:sfm-nontrivial-query} already improves upon $O(n^3)$ query algorithm of Queyranne \cite{Queyranne98}. However, if we want strongly polynomial (in $n$) query complexity, we get $O(n^3 \log^5 n)$ by plugging in $O(n^3 \log^2 n)$ query algorithm of \cite{LeeSW15}.
\end{remark}

\begin{proof}
This follows form \Cref{thm:sfm-nontrivial} and the result of \cite{LeeSW15}. The latter shows an algorithm for submodular function minimization that has query complexity $O(n^2 \log (nM))$. Call this algorithm as $\cA$. The query algorithm for finding non-trivial minimizer of $f$, denoted as $\cB$, works in $O(\log n)$ iterations where, in iteration $i$, \begin{itemize}
    \item $\cB$ sets $k = 1.5^i$,
    \item $\cB$ invokes \cref{thm:sfm-nontrivial} w.r.t. $\cA$.
\end{itemize}
In the end, $\cB$ outputs the non-trivial minimizer that has the smallest $f$ value among all $O(\log M)$ many minimizers.
\blikstad{Why $\log M$? should it not be $\log n$, as we only try powers of $k$ up to $n$ ($k$ is an approximation of the cardinality, and not the weight, if I understand correctly).} \sagnik{Done}

The correctness is immediate from that of \Cref{thm:sfm-nontrivial}. To see the complexity, let us analyze the complexity of one iteration. In iteration $i$, $\cB$ first makes $O(\log^2 n)$ calls to $\cA$ that takes $O(n^2 \log(nM) \log^2 n)$ queries. Then, $\cB$ also makes a number of calls to $\cA$, but on disjoint input space, $O(\log n)$ times. We claim that each such set of calls requires $O(n^2 \log(nM))$ queries. This follows simply because the calls are made on disjoint input space and the function $n^2 \log(nM)$ is a convex function. So the total query complexity of each iteration is $O(n^2 \log(nM) \log^2 n) + O(n^2 \log(nM) \log n) = O(n^2 \log(nM) \log^2 n)$. As there are $O(\log n)$ many iterations, the query upper bound follows.
\end{proof}

\section{Application II: Hypergraph min-cut} \label{sec:hg-mincut}

Finding min-cut in a hypergraph is an instance of SFM where the ground set is set the vertices $V$, and the submodular function is the cut function $C: 2^V \to {\mathbb R}$ such that $C(S) = \Abs{\Set{e=(u,v) \in E(G) \mid \text{either } u \in S \text{ or } v \in S}}$ \blikstad{that ``either \ldots or'' means \emph{exclusive or} can be made more clear}. It is not hard to see that this is a symmetric submodular function, and hence we are interested in finding the non-trivial minimizer of the the cut function.

Consider an algorithm $\cA$ which solves the $(s,t)$-min-cut on hypergraphs that has complexity $f(p, n)$ where $p$ is the representation size of, and $n$ is the number of vertices in the input hypergraph $G$. Using \Cref{thm:sfm-nontrivial}, we can infer the following:

\begin{lemma}
There is a randomized algorithm $\cB$ which solves hypergraph min-cut with $\polylog n$ calls to $\cA$ with complexity $\tO(f(O(n),O(p)))$ where $\tO(\cdot)$ hides polylogarithmic factors in terms of $n$.
\end{lemma}

Similar to the proof of \Cref{lem:sfm-nontrivial-query}, we invoke \Cref{thm:sfm-nontrivial} $O(\log n)$ times where, in the $i$-th iteration, we assume that $(1.5)^i$ is a constant approximation of the min-cut. As the size of the smaller side of the min-cut can be at most $n/2$, we need $O(\log n)$ iterations. 

In each iteration $i$, $\cB$ first makes $O(\log^2 n)$ many calls to $\cA$ corresponding to Step 1. The complexity of this step is $\tO(f(n,p))$. For Step 2, a better implementation is to make one call to $\cA$ where the input instance is all $(s,t)$-min-cut instances connected \textit{parallely} at $s$ and $t$: We identify source nodes of all $(s,t)$-min-cut instances by a single source node $s$ and, similarly, sink nodes of all $(s,t)$-min-cut instances by a single sink node $t$. The vertices and hyperedges of the   $(s,t)$-min-cut instances remain disjoint. From \Cref{thm:sfm-nontrivial}, we see that the number of vertices in this input instance is $O(n)$, even though the number of hyperedges can be large and one hyperedge can take part in many constituent $(s,t)$-min-cut instances. Nevertheless, note that in such an input instance, the contribution of any rank $r$ hyperedge to the representation size of the instance is $O(r)$. Hence the total representation size of the instance is $O(p)$ where $p$ is the representation size of the hypergraph $G$. Hence the complexity of Step 2 is $\tO(f(O(n), O(p)))$.

\section*{Acknowledgement}
We thank Joakim Blikstad for meticulous proof-reading. 
We thank Thatchaphol Saranurak for the early discussion on hypergraphs and the element connectivity problem.
We thank Jason Li and Debmalya Panigrahi for sharing the draft of \cite{LiP2020} and for many discussions. 
We also thank Calvin Beideman, Karthekeyan Chandrasekaran and Sorrachai Yingchareonthawornchai for discussions.
% This project has received funding from the European Research Council (ERC) under the
% European Union’s Horizon 2020 research and innovation programme under grant agreement No
% 715672. Authors are also supported by the Swedish Research Council (Reg. No. 2015-04659
% and 2019-05622).

\bibliography{biblio}

\begin{thebibliography}{LSW15}

\bibitem[CQ21]{chekuri2021}
Chandra Chekuri and Kent Quanrud.
\newblock Isolating cuts, (bi-)submodularity, and faster algorithms for global
  connectivity problems, 2021.

\bibitem[LP20]{LiP2020}
Jason Li and Debmalya Panigrahi.
\newblock Deterministic min-cut in poly-logarithmic max-flows.
\newblock In {\em {FOCS}}. {IEEE} Computer Society, 2020.

\bibitem[LSW15]{LeeSW15}
Yin~Tat Lee, Aaron Sidford, and Sam~Chiu{-}wai Wong.
\newblock A faster cutting plane method and its implications for combinatorial
  and convex optimization.
\newblock In {\em {FOCS}}, pages 1049--1065. {IEEE} Computer Society, 2015.

\bibitem[Que98]{Queyranne98}
Maurice Queyranne.
\newblock Minimizing symmetric submodular functions.
\newblock {\em Math. Program.}, 82:3--12, 1998.

\end{thebibliography}
\end{document}